\documentclass[11 pt]{article}
\usepackage{amsmath, todonotes, amsthm, amssymb,enumerate,verbatim,hyperref,mathrsfs}
\usepackage{color}
\usetikzlibrary{calc}
\usepackage{amsfonts}
\usepackage{graphicx}
\usepackage{mathtools}
\usepackage{pdfpages}
\usepackage{transparent}
\usetikzlibrary{calc,shapes}
\usepackage{authblk}
\usepackage[multiple]{footmisc}
\makeatletter
\renewcommand*\@fnsymbol[1]{\the#1}  
\makeatother

\usepackage{amsmath, amsthm,amssymb,bbm,enumerate,mathrsfs,mathtools}
\usepackage{a4wide}
\usepackage{tikz,xcolor,verbatim}
\usetikzlibrary{calc,shapes}
\usepackage{hyperref}
\usepackage{makecell}
\usepackage[numbers,sort&compress]{natbib}

\title{Graph Minors and the Linear Reducibility of Feynman Diagrams}
\author{Benjamin Moore\thanks{The authors would like to thank NSERC for financial support.  This work was partially completed while both authors were at Simon Fraser University.} \thanks{Email: brmoore@uwaterloo.ca}
}
\author{Karen Yeats\thanks{Email: kayeats@uwaterloo.ca}}
\affil{Department of Combinatorics and Optimization, University of Waterloo}
\date{}

\newtheoremstyle{case}{}{}{\normalfont}{}{\itshape}{:}{ }{}

\newtheorem{theorem}{Theorem}

\newtheorem{proposition}[theorem]{Proposition}
\newtheorem{definition}[theorem]{Definition}

\newtheorem{Conjecture}[theorem]{Conjecture}
\newtheorem{observation}[theorem]{Observation}

\newtheorem{lemma}[theorem]{Lemma}

\newtheoremstyle{case}{}{}{\normalfont}{}{\itshape}{\normalfont:}{ }{}
\theoremstyle{case}

\begin{document}

\maketitle

\begin{abstract}
We look at a graph property called reducibility which is closely related to a condition developed by Brown to evaluate Feynman integrals. We show for graphs with a fixed number of external momenta, that reducibility with respect to both Symanzik polynomials is graph minor closed. We also survey the known forbidden minors and the known structural results.  This gives some structural information on those Feynman diagrams which are reducible.
\end{abstract}

\section{Introduction}

In recent years, there has been a large amount of research on both the mathematical and practical sides of calculating Feynman integrals using multiple polylogarithms \cite{FrancisBig, panzerphdthesis, olivercensus,Panzertalk,Schnetztalk}.  In \cite{FrancisBig}, Brown gave a sufficient condition, called \textit{linear reducibility}, for evaluating Feynman integrals in parametric form iteratively using multiple polylogarithms. Using this condition, Brown showed that if a graph has vertex width less than three, then the corresponding Feynman period evaluates to  multiple zeta values. To determine if a Feynman integral is linearly reducible, one does not have to look at the integral at all, instead needing only to focus on a few graph polynomials, namely the Symanzik polynomials. We quickly review the definitions of the Symanzik polynomials. 

Let $G$ be a graph. To each edge $e \in E(G)$ we associate a Schwinger parameter $\alpha_{e}$. Let $\mathcal{T}$ denote the set of spanning trees of $G$. Then the \textit{first Symanzik polynomial} of $G$ is

\[\Psi_{G} = \sum_{T \in \mathcal{T}} \prod_{e \not \in E(T)} \alpha_{e}. \]

We would also like to allow kinematics. To do so, to each edge $e$ we associate an edge weight $m_{e} \in \mathbb{R}$, the \textit{mass} of $e$. Also, we will allow our graph to have \textit{external edges}, as in, edges with exactly one endpoint. The external edges have associated external momenta, but it will be more convenient to associate the external momenta to the vertices with external edges.  Formally for each $v \in V(G)$, we associate a vector $\rho_{v} \in \mathbb{R}^{4}$ (with Minkowski signature) called the \textit{external momentum at $v$} which will be $0$ for vertices without external edges.  This lines up neatly with the graph theory; the vertices with external edges become the root vertices in rooted forbidden minor results.  For the results of Section~\ref{sec minors}, we will restrict to on-shell external momenta.  Let $G$ be a graph and $H$ a subgraph of $G$. We will say the momentum flowing into $H$ is 

\[\sum_{v \in V(H)} \rho_{v}.\]

Given a graph $G$, a \textit{spanning $2$-forest of $G$} is an unordered pair $(T_{1},T_{2})$ where $T_{1}$ and $T_{2}$ are trees such that $V(T_{1}) \cup V(T_{2}) = V(G)$ and $V(T_{1}) \cap V(T_{2}) = \emptyset$. Then, letting $\mathcal{T}$ be the set of spanning $2$-forests of $G$, the \textit{second Symanzik polynomial} is

 \[\Phi_{G} = \sum_{\mathclap{(T_{1},T_{2}) \in \mathcal{T}}} (\rho^{T_{1}})^{2}  \prod_{\mathclap{e \not \in E(T_{1} \cup T_{2})}} \alpha_{e} + \Psi_{G}\sum_{i =1}^{|E(G)|} \alpha_{i}m_{i}^{2}.\]

Here $(\rho^{T_{1}})^{2}$ is taken to mean the Minkowski norm squared.  We note some authors let the second Symanzik polynomial be just the terms depending on the external momenta \cite{FrancisSmall,Brownsurveypaper}. Following \cite{Bognerpaper}, if $G$ is a disconnected graph, then we say $\Phi_{G} = 0$ and $\Psi_{G} = 0$. This is a reasonable convention; if we start with a connected graph, we want to view deleting edges as restricting the number of spanning trees and spanning $2$-forests until there are none. Of course, a disconnected graph can have a spanning $2$-forest, however in such a case the spanning $2$-forest is really just two spanning trees of the connected components, and thus not what we would like to consider. 

At a high level, a Feynman integral being linearly reducible means that starting with the Symanzik polynomials, there exists an ordering such that one can iteratively integrate the integral in parametric form such that polynomials which appear can be dealt with by multiple polylogarithms.  Therefore, we can talk about sets of polynomials being linearly reducible. We note that such an ordering does not always occur, and there are examples of Feynman integrals which require a larger class of functions than multiple polylogarithms \cite{nonmultiplezetavalue}.

Surprisingly, Brown showed that for the first Symanzik polynomial, linear reducibility is graph minor closed \cite{FrancisBig}, as in, given a reducible graph, the deletion or contraction of any edge results in a reducible graph. Then, by the celebrated well-quasi-ordering theorem of Robertson and Seymour \cite{WQOtheorem}, linear reducibility  for the first Symanzik polynomial is characterized by a finite set of forbidden minors. Due to this, there have been attempts to understand linear reducibility (and various related notions) using graph minors \cite{Crumppaper,Bognerpaper,FrancisBig,Martin}. One such attempt by Bogner and L\"uders was to extend the graph minor closed result to both Symanzik polynomials.  They showed that a  stronger condition called Fubini reducibility\footnote{Bogner and L\"uders occasionally used the term linear reducibility, however their version of reducibility is the same as the one outlined in \cite{FrancisSmall} which Brown calls Fubini reducibility.} is graph minor closed for graphs with a fixed number of external momenta \cite{Bognerpaper}. Unfortunately there are a few flaws in their argument. For one, the definition of graph minors includes the deletion of isolated vertices, for which their claim fails. This is because deleting an isolated vertex can cause a disconnected graph to become connected, which can cause drastic changes in $\Phi$ and $\Psi$ and result in a non-Fubini reducible graph. This is not too much of a problem, as when one restricts to connected graphs, one never needs to delete isolated vertices when finding minors. However, their proof does not follow even in the connected case.

In this paper, we correct and generalize the Bogner and L\"uders result to ``compatibility graph reducibility". Here we note compatibility graph reducibility is a stronger condition than the notion of linear reducibility in \cite{FrancisBig} (as in, if a set is compatibility graph reducible, then it is linearly reducible), but is weaker than the notion of Fubini reducibility used by Bogner and L\"uders (as in, Fubini reducibility implies compatibility graph reducibility). Additionally, we list the known forbidden minors for graphs with four external on-shell momenta (with one new forbidden minor included) and outline some structural results for graphs excluding these minors.

\section{The Compatibility Graph Reduction}

Here we review the compatibility graph reduction outlined in \cite{FrancisBig}. Before jumping into the definition, we pause to give some motivation and intuition. The compatibility graph reduction is an algorithm which simulates what polynomials would appear at each stage of iteratively integrating a parametric Feynman integral under some order, by only acting on the polynomials. At each step, the algorithm checks if the techniques in \cite{FrancisBig} are admissible for the integration, or stops if it finds an obstruction. While compatibility graph reducibility is a sufficient condition for the evaluation of Feynman integrals, in general one can ask if an arbitrary set of polynomials with rational coefficients is compatibility graph reducible. We will first give a simpler variant of the algorithm, and use that to give the full definition.

Let $S$ be a set of polynomials in the polynomial ring $Q[\alpha_1,\alpha_2,\ldots, \alpha_r]$.  Let $\sigma$ be a permutation of $\{\alpha_1,\ldots,\alpha_r\}$. Let $C_{S}$ be a complete graph on $|S|$ vertices, where we consider the vertex set of $C_{S}$ to be the set of polynomials in $S$. Suppose we are at the $k^{th}$ iteration of the algorithm. If $k \geq 2$, then  we have a set of polynomials with rational coefficients $S_{(\sigma(1),\ldots,\sigma(k-1))} =  \{f_{1},\ldots,f_{n}\}$ and compatibility graph $C_{(\sigma(1),\ldots,\sigma(k-1))}$ (here and throughout we abuse notation and let $\sigma(i) = \alpha_{\sigma(i)}$). Here we let the vertices of $C_{(\sigma(1),\ldots,\sigma(k-1))}$ be labelled by the polynomials from $S_{(\sigma(1),\ldots,\sigma(k-1))}$, and we define the edge set below. Otherwise $k=1$ and we use $S$ and $C_{S}$. We then do the following:

\begin{enumerate}
\item{If there is a polynomial $f \in S_{(\sigma(1),\ldots,\sigma(k))}$ such that $f$ is not linear in $\sigma(k)$ (here and throughout linear means the degree of $f$ in $\sigma(k)$ is at most one, so in particular $f$ may be constant in $\sigma(k)$), we end the algorithm, otherwise continue.}

\item{Then for all $i \in \{1,\ldots,n\}$, given a polynomial $f_i \in S_{(\sigma(1),\ldots,\sigma(k))}$, we write $f_{i} = g_{i}\sigma(k) + h_{i}$, where $g_{i} =  \frac{\partial f_{i} }{\partial \sigma(k)}$ and $h_{i} = f_{i}|_{\sigma(k)=0}$.}

\item{Let $S^{1} = \{g_{i} | i \in \{1,\ldots,n\} \}$. Let $S^{2} = \{h_{i} | i \in \{1,\ldots,n\} \}$. Let $S^{3} = \{g_{i}h_{j}-h_{i}g_{j}| i,j \in \{1,\ldots,n\}, i \neq j  \text{ where }  f_{i}f_{j} \in E(C_{(\sigma(1),\ldots,\sigma(k-1))}) \}$. Let $S^{4} = S^{1} \cup S^{2} \cup S^{3}$. }

\item{Let $\tilde{S}$ be the set of irreducible factors over $\mathbb{Q}$ of polynomials in $S^{4}$.}

\item{Let $S_{(\sigma(1),\ldots,\sigma(k))} = \tilde{S}$ and construct a new compatibility graph $C_{(\sigma(1),\ldots,\sigma(k))}$ with some rule set (see below).}

\item{ Repeat the above steps with $S_{(\sigma(1),\ldots,\sigma(k))}$ in place of $S_{(\sigma(1),\ldots,\sigma(k-1))}$ and $C_{(\sigma(1),\ldots,\sigma(k))}$ in place of $C_{(\sigma(1),\ldots,\sigma(k-1))}$.}
\end{enumerate}

Now we show how to construct the compatibility graphs, $C_{(\sigma(1),\ldots,\sigma(k))}$. Recall, the vertex set of $C_{(\sigma(1),\ldots,\sigma(k))}$ is the set of polynomials in $S_{(\sigma(1),\ldots \sigma(k-1))}$. To determine the edges of $C_{(\sigma(1),\ldots,\sigma(k))}$ we associate a set of $2$-tuples to each vertex of $C_{(\sigma(1),\ldots,\sigma(k))}$ to keep track of how the polynomials were created. Let $m \in V(C_{(\sigma(1),\ldots,\sigma(k))})$. 

\begin{itemize}
\item{If $m$ is an irreducible factor of a polynomial $g_{i} \in S^{1}$ we associate the $2$-tuple $\{0,i\}$ with $m$.}

\item{ If $m$ is an irreducible factor of some polynomial $h_{i} \in S^{2}$ then we associate the $2$-tuple $\{i,\infty\}$ with $m$. Additionally, if $h_{i} = f_{i}$ where $f_{i} \in S_{(\sigma(1),\ldots \sigma(k-1))}$ then we associate the $2$-tuple $\{0,i\}$ to $m$ as well as $\{i,\infty\}$.}

\item{If $m$ is the irreducible polynomial of some polynomial $g_{i}h_{j}-h_{i}g_{j} \in S^{3}$, then we associate the $2$-tuple $\{i,j\}$ to $m$.}
\end{itemize}

Now let $m,n \in V(C_{(\sigma(1),\ldots,\sigma(k))})$. The edge $mn \in E(C_{(\sigma(1),\ldots,\sigma(k))})$ if and only if there exists a $2$-tuple associated with $m$ and a $2$-tuple associated with $n$ such that their intersection is non-empty. 

 As an example, suppose a polynomial $f$ is associated with the $2$-tuples $\{1,\infty\}$ and $\{3,4\}$, and polynomial $g$ is associated with the $2$-tuples $\{2,\infty\}$ and $\{1,3\}$. Then since $\{3,4\} \cap \{1,3\} = \{3\}$, the polynomials $f$ and $g$ are compatible. However if $g$ instead was associated with the $2$-tuples $\{5,6\}$ and $\{0,2\}$, then as all pairwise intersections are empty, $f$ and $g$ would not be compatible. We note that this technique does not always result in a complete graph (and thus we consider fewer polynomials in the reduction algorithm). We refer the reader to \cite{FrancisBig} for examples of compatibility graphs which are not the complete graphs when starting with the first Symanzik polynomial. We remark that there is a similar notion of compatibility graphs developed by Panzer in \cite{panzerphdthesis}. 
 
 Now we define the full compatibility graph reduction. Let $S$ be a set of polynomials in the polynomial ring $Q[\alpha_1,\alpha_2,\ldots, \alpha_r]$. Initialize $C_{S}$ to be the complete graph on $|S|$ vertices and let $\sigma$ be a permutation of $\{\alpha_1,\ldots,\alpha_r\}$.

We define $S_{[\sigma(1)]} = S_{(\sigma(1))}$ and $C_{[\sigma(1)]} = C_{(\sigma(1))}$, where we obtain $S_{(\sigma(1))}$ and $C_{(\sigma(1))}$ using the algorithm outlined at the start of the section. Additionally,  we define $S_{[\sigma(1),\sigma(2)]} = S_{(\sigma(1),\sigma(2))} \cap S_{(\sigma(2),\sigma(1))}$ and $C_{[\sigma(1),\sigma(2)]}$ to be a graph where $fg \in E(C_{[\sigma(1),\sigma(2)]})$ if and only if $fg \in E(C_{(\sigma(1),\sigma(2))}) \cap E(C_{(\sigma(2),\sigma(1))})$. Then we can inductively define the sets $S_{[\sigma(1),\ldots,\sigma(k)]}$ by 

\begin{center}
\[S_{[\sigma(1),\ldots,\sigma(k)]} = \bigcap_{1 \leq i \leq k} S_{[\sigma(1),\ldots,\hat{\sigma(i)},\ldots,\sigma(k)](\sigma(i))},  \]
\end{center}

and we define $C_{[\sigma(1),\ldots,\sigma(k)]}$ to be the compatibility graph for $S_{[\sigma(1),\ldots,\sigma(k)]}$ such that $fg \in C_{[\sigma(1),\ldots,\sigma(k)]}$ if and only if

 \[fg \in  \bigcap_{1 \leq i \leq k} E(C_{[\sigma(1),\ldots,\hat{\sigma(i)},\ldots,\sigma(k)](\sigma(i))}).\]
 
Here we make a slight technical note. In the above notions of intersecting sets of polynomials, if two polynomials differ by a constant prefactor, we do not remove them from the intersection. Now finally we can define what it means for a set to be compatibility graph reducible. 

\begin{definition}
Let $S$ be a set of polynomials in the polynomial ring $\mathbb{Q}[\alpha_{1},\ldots,\alpha_{r}]$.  Let $\sigma$ be a permutation of $\{\alpha_{1},\ldots,\alpha_{r}\}$. We say $S$ is \emph{compatibility graph reducible with respect to $\sigma$} if for all $ 1 \leq i \leq r-1$, all polynomials in the set $S_{[\sigma(1),\ldots,\sigma(i)]}$ are linear in $\sigma(i+1)$. If there exists a permutation $\sigma$ such that $S$ is compatibility graph reducible with respect to $\sigma$, we say that $S$ is \emph{compatibility graph reducible}. Given a graph $G$, and $S \subseteq \{\Psi_{G},\Phi_{G}\}$, if $S$ is compatibility graph reducible, then we say $G$ is \emph{compatibility graph reducible with respect to $S$.}
\end{definition}

We pause to remark on what happens when a variable in $\mathbb{Q}[\alpha_{1},\ldots,\alpha_{r}]$ does not appear in the set $S$. In this case, let $\alpha$ be the variable, and apply one step of the reduction algorithm to $S$ and $\alpha$. We then obtain the set $S_{[\alpha]}$ which is the set of irreducible factors of $S$ over $\mathbb{Q}$, and the resulting compatibility graph is a complete graph. In some sense, we can always assume every variable appears in some polynomial of $S$. To see this, suppose we pick an ordering of our variables where all the variables which do not appear come first. Then as the initial compatibility graph is complete, applying the reduction algorithm to those variables simply makes the set of polynomials irreducible over $\mathbb{Q}$. We show later (Lemma \ref{irreduciblepolynomials}) that this does not affect reducibility. 

We note the difference between Fubini reducibility used by Bogner and L\"uders and compatibility graph reducibility, is that in Fubini reducibility, at each step every compatibility graph is a complete graph. There are graphs which are compatibility graph reducible but not Fubini reducible. For example, consider the $4$-cycle $C_{4}$, with four on-shell external momenta. Then $C_{4}$ is not Fubini reducible with respect to $\Phi$ and $\Psi$, but is compatibility graph reducible \cite{Martin}. For brevity, we will refer to compatibility graph reducibility as just reducibility.

\section{Properties of Reducible sets}

In this section, we prove the following result:

\begin{theorem}
\label{fullclaim}
Let $S = \{P_{1},\ldots,P_{N}\}$ be a set of polynomials which is reducible in the order $(\alpha_{1},\ldots,\alpha_{n})$. Fix $l \in \{1,\ldots,n\}$. Let $lc(P)$ denote the leading coefficient of polynomial $P$ with respect to $\alpha_{l}$.   Let $S' = \{lc(P)| P \in S\}$ and $S'' = \{P_{1}|_{\alpha_{l}=0}, \ldots, P_{N}|_{\alpha_{l}=0}\}$. Then both $S'$ and $S''$ are reducible with the order $(\alpha_{1},\ldots, \alpha_{n})$.
\end{theorem}

This will essentially give us that reducibility is graph minor closed for the Symanzik polynomials. We start this section by noting that reducibility is well behaved under subsets.


\begin{lemma}
\label{subsetlemma}
Let $S = \{f_1,\ldots,f_n\}$ be a set of polynomials in the polynomial ring $\mathbb{Q}[\alpha_{1},\ldots,\alpha_{r}]$ and let $\sigma$ be a permutation of $\{\alpha_{1},\ldots,\alpha_{r}\}$. Suppose $S$ is reducible with respect to $\sigma$. Then any subset $L \subseteq S$ is reducible with respect to $\sigma$.  
\end{lemma}

\begin{proof}

Consider the sequence of sets and compatibility graphs \[(S,C_{S}), (S_{[\sigma(1)]},C_{[\sigma(1)]}),\ldots,(S_{[\sigma(1),\ldots,\sigma(r-1)]},C_{[\sigma(1),\ldots,\sigma(r-1)]}).\]

 We claim that that for all $i \in \{0,1,\ldots,r-1\}$, the set $L_{[\sigma(1),\ldots,\sigma(i)]} \subseteq S_{[\sigma(1),\ldots,\sigma(i)]}$ and that $C_{L_{[\sigma(1),\ldots,\sigma(i)]}}$ is a subgraph of $C_{S_{[\sigma(1),\ldots,\sigma(i)]}}[L_{[\sigma(1),\ldots,\sigma(i)]}]$ (that is, a subgraph of the graph induced by the polynomials $L_{[\sigma(1),\ldots,\sigma(i)]}$ in the graph $C_{S_{[\sigma(1),\ldots,\sigma(i)]}}$). When $i=0$ we consider the sets $(S,C_{S})$ and $(L,C_{L})$. We proceed by induction on $i$.
 
The base case follows trivially. Now consider the set $L_{[\sigma(1),\ldots,\sigma(i-1)]}$ and the compatibility graph $L_{C_{[\sigma(1),\ldots,\sigma(i-1)]}}$. By induction we have $L_{[\sigma(1),\ldots,\sigma(i-1)]} \subseteq S_{[\sigma(1),\ldots,\sigma(i-1)]}$ and that $C_{L_{[\sigma(1),\ldots,\sigma(i-1)]}}$ is a subgraph of $C_{S_{[\sigma(1),\ldots,\sigma(i-1)]}}[L_{[\sigma(1),\ldots,\sigma(i-1)]}]$. We will now apply one step of the reduction algorithm to these sets. 

We claim that for all $1 \leq j \leq i$, we have $ L_{[\sigma(1),\ldots,\hat{\sigma(j)},\ldots,\sigma(i)](\sigma(j))} \subseteq S_{[\sigma(1),\ldots,\hat{\sigma(j)},\ldots,\sigma(i)](\sigma(j))}$ when the set  $S_{[\sigma(1),\ldots,\hat{\sigma(j)},\ldots,\sigma(i)](\sigma(j))}$ exists. We know that $S$ is reducible with respect to $\sigma$, so there exists a $j \in \{1,\ldots,i\}$ such that $S_{[\sigma(1),\ldots,\hat{\sigma(j)},\ldots,\sigma(i)](\sigma(j))}$ exists. Fix any such $j \in \{1,\ldots,i\}$. Then by induction we have 

\[L_{[\sigma(1),\ldots,\hat{\sigma(j)},\ldots,\sigma(i)]} \subseteq S_{[\sigma(1),\ldots,\hat{\sigma(j)},\ldots,\sigma(i)]},\] and $C_{L_{[\sigma(1),\ldots,\hat{\sigma(j)},\ldots,\sigma(i)]}}$ is a subgraph of $C_{S_{[\sigma(1),\ldots,\hat{\sigma(j)},\ldots,\sigma(i)]}}[L_{[\sigma(1),\ldots,\hat{\sigma(j)},\ldots,\sigma(i)]}].$

Note as $S_{[\sigma(1),\ldots,\hat{\sigma(j)},\ldots,\sigma(i)](\sigma(j))}$ exists, all polynomials in $L_{[\sigma(1),\ldots,\hat{\sigma(j)},\ldots,\sigma(i)]}$ are linear in $\sigma(j)$. We will denote the set $L^{4}$ to be the set obtained in the third step of the reduction algorithm when starting with $L_{[\sigma(1),\ldots,\hat{\sigma(j)},\ldots,\sigma(i)]}$, and $S^{4}$ will be the set obtained in the second step of the reduction algorithm applied to $S_{[\sigma(1),\ldots,\hat{\sigma(j)},\ldots,\sigma(i)]}$.  Then $L^{4} \subseteq S^{4}$ as $L_{[\sigma(1),\ldots,\hat{\sigma(j)},\ldots,\sigma(i)]} \subseteq S_{[\sigma(1),\ldots,\hat{\sigma(j)},\ldots,\sigma(i)]}$ and as two polynomials in $L_{[\sigma(1),\ldots,\hat{\sigma(j)},\ldots,\sigma(i)]}$ are compatible then they are compatible in $ S_{[\sigma(1),\ldots,\hat{\sigma(j)},\ldots,\sigma(i)]}$.

 Let $\tilde{S}$ and $\tilde{L}$ be the sets obtained from step four of the reduction algorithm in for $S_{[\sigma(1),\ldots,\hat{\sigma(j)},\ldots,\sigma(i)]}$ and $L_{[\sigma(1),\ldots,\hat{\sigma(j)},\ldots,\sigma(i)]}$ respectively. Then as $L^{4} \subseteq S^{4}$, we have that  $\tilde{L} \subseteq \tilde{S}$.
 Now consider the compatibility graph $C_{L_{[\sigma(1),\ldots,\hat{\sigma(j)},\ldots,\sigma(i)](\sigma(j))}}$. Let $m$ and $n$ be adjacent in the compatibility graph. Then there is an $2$-tuple associated to $m$ and a $2$-tuple associated to  $n$ such that their intersection is non-empty. By construction, these $2$-tuples are derived from some polynomials in $L_{[\sigma(1),\ldots,\hat{\sigma(j)},\ldots,\sigma(i)](\sigma(j))}$. But we know that $L_{[\sigma(1),\ldots,\hat{\sigma(j)},\ldots,\sigma(i)](\sigma(j))} \subseteq S_{[\sigma(1),\ldots,\hat{\sigma(j)},\ldots,\sigma(i)](\sigma(j))}$ 
 and thus the polynomials which gave $m$ and $n$ the associated $2$-tuples in $L_{[\sigma(1),\ldots,\hat{\sigma(j)},\ldots,\sigma(i)](\sigma(j))}$ exist in $S_{[\sigma(1),\ldots,\hat{\sigma(j)},\ldots,\sigma(i)](\sigma(j))}$. Thus $mn \in E(C_{S_{[\sigma(1),\ldots,\hat{\sigma(j)},\ldots,\sigma(i)](\sigma(j))}})$. Therefore we have 
 $C_{L_{[\sigma(1),\ldots,\hat{\sigma(j)},\ldots,\sigma(i)](\sigma(j))}}$ is a subgraph of the graph induced by $L_{[\sigma(1),\ldots,\hat{\sigma(j)},\ldots,\sigma(i)](\sigma(j))}$ in $C_{S_{[\sigma(1),\ldots,\hat{\sigma(j)},\ldots,\sigma(i)](\sigma(j))}}$.  Therefore,
 
 \[L_{[\sigma(1),\ldots,\sigma(i)]} = \bigcap_{1 \leq j \leq k} L_{[\sigma(1),\ldots,\hat{\sigma(j)},\ldots,\sigma(i)](\sigma(j))} \subseteq \bigcap_{1 \leq j \leq k} S_{[\sigma(1),\ldots,\hat{\sigma(j)},\ldots,\sigma(i)](\sigma(j))} = S_{[\sigma(1),\ldots,\sigma(i)]},  \]
 
 and $C_{L_{[\sigma(1),\ldots,\sigma(i)]}}$ is a subgraph of $C_{S_{[\sigma(1),\ldots\sigma(i)]}}[L_{[\sigma(1),\ldots,\sigma(i)]}]$, completing the claim.
\end{proof}

In proving Theorem \ref{fullclaim}, we will move between some set of polynomial $S$ and the set of all irreducible polynomials of $S$ with respect to $\mathbb{Q}$. We now prove a lemma that shows we can essentially move between these two sets without impacting reducibility.

\begin{lemma}
\label{irreduciblepolynomials}
Let $S = \{f_{1},\ldots,f_{n}\}$ be a set of polynomials in the polynomial ring $\mathbb{Q}[\alpha_{1},\ldots,\alpha_{r}]$. Let $\sigma$ be a permutation of $\{\alpha_{1},\ldots,\alpha_{r}\}$. Let $S^{I}$ be the set of irreducible factors of $S$ with rational coefficients.  If $S$ is reducible with respect to $\sigma$ then $S^{I}$ is reducible with respect to $\sigma$. Furthermore, if $S^{I}$ is reducible with respect to $\sigma$,  and all polynomials in $S$ are linear in $\sigma(1)$, then $S$ is reducible with respect to $\sigma$. 
\end{lemma}

\begin{proof}

Suppose $S$ is reducible with respect to $\sigma$. Let 
$p \in S$ and $p = p_{1}p_{2}$ where $p_{1},p_{2}$ are polynomials in $\mathbb{Q}[\alpha_{1},\ldots,\alpha_{r}]$. Let $L = S \setminus \{p\} \cup \{p_{1},p_{2}\}$. We claim $L$ is reducible with respect to $\sigma$. Note this  proves the first statement of the lemma as one can repeatedly apply this fact. We will show that $S_{[\sigma(1)]} = L_{[\sigma(1)]}$ and that their compatibility graphs are isomorphic. We consider two cases.

\textbf{Case 1:} Suppose $\deg(p,\sigma(1)) =0$. Then  $\deg(p_{i},\sigma(1)) =0$ for $i \in \{1,2\}$. We perform the first iteration of the reduction algorithm on $S$ and $L$. Let $L^{4}$ and $S^{4}$ be the sets obtained in step three of the reduction algorithm for $\sigma(1)$ for $L$ and $S$ respectively. As $\deg(p,\sigma(1)) = 0$, we have that $p_{1},p_{2} \in L^{4}$ and $p \in S^{4}$. Notice that the polynomial $g_{i}h_{j} - g_{j}h_{i} = pg_{i}$ when the $h_{j}$, $g_{j}$ polynomials are obtained from $p$, as then $g_{j}$ is $0$. Similarly, when $h_{j}$, $g_{j}$ is obtained from $p_{1}$ or $p_{2}$, we get that $g_{i}h_{j} - g_{j}h_{i}$ is  $p_{1}g_{i}$ or $p_{2}g_{i}$ respectively. Then after factoring, we have  $\tilde{L} = \tilde{S}$, and so $L_{[\sigma(1)]} = S_{[\sigma(1)]}$. Thus it suffices to show $C_{L_{(\sigma(1))}}$ is isomorphic to $C_{S_{(\sigma(1))}}$.

As $L = S \setminus \{p\} \cup \{p_{1},p_{2}\}$, the graphs $C_{L_{[\sigma(1)]}}$ and $C_{S_{[\sigma(1)]}}$ have the same vertex set, and we may restrict our attention to compatibilities caused by $p$, $p_{1}$ and $p_{2}$. Notice the irreducible factors of $p$, $p_{1}$ and $p_{2}$ all contain an associated $2$-tuple containing $\infty$ and $0$, thus we may restrict our attention to $2$-tuples not containing $\infty$ or $0$.

Suppose an irreducible factor of $p$, say $p'$, is adjacent to a vertex $g'$ in $S_{C_{\sigma(1)}}$, where $g'$ is an irreducible factor of some $g_{i}$. Since we are assuming that the compatibility did not arise from a $2$-tuple  containing $\{0\}$ or $\{\infty\}$ we may assume that compatibilities comes from the polynomial $pg_{j}$. Then without loss of generality, $p'$ is an irreducible factor of  $p_{1}$ and so $p_{1}g_{j}$ would provide the desired $2$-tuple for compatibility. The same argument works in the other direction as well, so we have that $C_{L_{[\sigma(1)]}}$ is isomorphic to $C_{S_{[\sigma(1)]}}$. As $S_{[\sigma(1)]} = L_{[\sigma(1)]}$ and their compatibility graphs are the same, since $S$ is reducible, $L$ is reducible.

\textbf{Case 2:} Suppose $\deg(p,\sigma(1)) =1$. Notice that as $p$ is linear in $\sigma(1)$, exactly one of $p_{1}$ or $p_{2}$ is linear in $\sigma(1)$. Without loss of generality, we assume $\deg(p_{1},\sigma(1)) =1$.
 
 As before let $S^{4}, \tilde{S}$ and $L_{4},\tilde{L}$ denote the sets obtained from the third and fourth step of the reduction for $S$ 
 and $L$ respectively. Let $p_{1} = g_{1}\sigma(1) +h_{1}$. Then $p = p_2(g_{1}\sigma(1) + h_{1})$. Thus  $p_{2},h_{1},g_{1} \in L^{4}$ , and $p_{2}g_{1}, p_{2}h_{1} \in S^{4}$.  Notice that the polynomial $g_{i}h_{j} - h_{i}g_{j} = p_2(g_{1}h_{j}-h_{1}g_{j})$ when the $g_{i}, h_{i}$ polynomials are obtained from $p$. Thus  $S_{[\sigma(1)]}$ will contain all the irreducible factors obtained from $p_{1}$ and $p_{2}$. Similarly, $L_{[\sigma(1)]}$ will contain all the irreducible factors obtained from $p$, so $S_{[\sigma(1)]} = L_{[\sigma(1)]}$. 
 
For the compatibility graphs, as before all irreducible factors of $p$ are adjacent in $C_{S_{[\sigma(1)]}}$ as they share an $2$-tuple from $p$. In $C_{L_{[\sigma(1)]}}$, all irreducible factors of either $p_{1}$ or $p_{2}$ are adjacent. This follows since $\deg(p_{2},\sigma(1))=0$, all irreducible factors of $p_{2}$ have a $2$-tuple with $0$ and a $2$-tuple with $\infty$, and every irreducible factor of $p_{1}$ has a $2$-tuple containing either $0$ or $\infty$. If $f$ and $g$ are adjacent in $C_{S_{[\sigma(1)]}}$ and the $2$-tuples which make them adjacent came from a polynomial $p_{2}(g_{1}h_{j}-h_{1}g_{j})$, then the desired $2$-tuples exist for $C_{L_{[\sigma(1)]}}$ as the polynomials $p_{2}g_{1} \in L^{4}$ and $g_{1}h_{j}-h_{1}g_{j} \in L^{4}$. A similar statement holds for two adjacent polynomials in $C_{L_{[\sigma(1)]}}$. Therefore one can see that $C_{L_{[\sigma(1)]}}$ is isomorphic to $C_{S_{[\sigma(1)]}}$, and thus as $S$ is reducible with respect to $\sigma$, $L$ is reducible with respect to $\sigma$. Notice for the partial converse, the same argument works, as the only obstruction is if the reduction algorithm stops immediately.
\end{proof}

We remark that in the above proof we did not rely on the fact that the initial compatibility graph is complete, so one can apply this argument in the ``middle" of a reduction if desired. Now we are in position to prove Theorem \ref{fullclaim} (split into two theorems as the arguments are slightly different).

\begin{theorem}
\label{reducibilityevalutation}
Let $S = \{P_{1},\ldots,P_{N}\}$ be a set of polynomials which is reducible in the order $(\alpha_{1},\ldots, \alpha_{n})$. Fix some $l \in \{1,\ldots,n\}$. Then the set $S^{l} = \{P_{1}|_{\alpha_{l} =0}, \ldots, P_{N}|_{\alpha_{l}=0}\}$ is reducible in the order $(\alpha_{1},\ldots, \alpha_{n})$.
\end{theorem}

\begin{proof}
Let $S$ and $S^{l}$ be a counterexample with $n$ minimized, as in, $S$ is reducible with order $(\alpha_{1},\ldots,\alpha_{n})$, but $S^{l}$ is not reducible with order $(\alpha_{1},\ldots,\alpha_{n})$. 

First suppose that $l= 1$. Then for all polynomials $P \in S$, we have that $\deg(P,\alpha_{l}) \leq 1$. Then notice that $S^{l} = S^{2}$ where $S^{2}$ is the set obtained by applying the reduction algorithm to $S$ for $\alpha_{l}$. Then $S^{l}_{[\alpha_{l}]} \subseteq S_{[\alpha_{1}]}$. Then by Lemma \ref{subsetlemma}, we have that $S^{l}_{[\alpha_{l}]}$ is reducible with order $(\alpha_{2},\ldots,\alpha_{n})$, which implies that $S^{l}$ is reducible with order $(\alpha_{1},\ldots,\alpha_{n})$, a contradiction.

 Therefore we assume that $l \neq 1$ and consider one step of the reduction algorithm. As $S$ is reducible with order $(\alpha_{1},\ldots,\alpha_{n})$, we have that $S_{[\alpha_{1}]}$ is reducible with order $(\alpha_{2},\ldots,\alpha_{n})$. Consider the set $S^{l'}_{[\alpha_{1}]} = \{f|_{\alpha_{l}=0} | f \in S_{[\alpha_{1}]} \}$.  Then since $S$ is a minimal counterexample with respect to $n$, we have that $S^{l'}_{[\alpha_{1}]}$ is reducible with order $(\alpha_{2},\ldots,\alpha_{n})$. Additionally, by Lemma \ref{irreduciblepolynomials}, the set of irreducible polynomials of  $S^{l'}_{[\alpha_{1}]}$  is reducible with order $(\alpha_{2},\ldots,\alpha_{n})$. Notice from the definitions we have:
\begin{align*}
S^{l'}_{[\alpha_{1}]} &= \{f|_{\alpha_{l}=0} | f \in S_{[\alpha_{1}]} \} 
\\&=
\{f|_{\alpha_{l}=0} | f \in \text{irreducible factors of } S^{4} \}.
\end{align*}

 For notational convenience, we will say $S^{jl}$ will be the set $S^{j}$ obtained from the reduction algorithm by starting with $S^{l}$ for $j \in \{1,2,3,4\}$. Suppose we have a polynomial $f$, and $f = f_{1}f_{2}$ for some polynomials $f_{1}$ and $f_{2}$. Then fix any variable $\alpha$ and notice that  $f|_{\alpha =0} = f_{1}|_{\alpha = 0}f_{2}|_{\alpha =0}$. This follows since any term which contains $\alpha$ in $f$ is generated by a pair of terms in $f_{1}$ and $f_{2}$, where at least one of these terms contains $\alpha$. Now, notice that $S^{l}_{[\alpha_{1}]}$ exists and,

\begin{align*}
S^{l}_{[\alpha_{1}]} &= \text{irreducible factors of } S^{4l}
 \\ &=
\text{irreducible factors of } (S^{1l} \cup S^{2l} \cup S^{3l}) 
\\ &=
\text{irreducible factors of } (\{ \frac{\partial f}{\partial \alpha_{1}} | f \in S^{l} \} \cup    \{f|_{\alpha_{1}=0} | f \in S^{l} \} \cup \\& \qquad \{\frac{\partial f_{1}}{\partial \alpha_{1}}f_{2}|_{\alpha_{1} = 0} - \frac{\partial f_{2}}{\partial \alpha_{1}}f_{1}|_{\alpha_{1}=0} | f_{1},f_{2} \in S^{l} \}) 
\\&=
\text{irreducible factors of } (\{ \frac{\partial f|_{\alpha_{l}=0}}{\partial \alpha_{1}}  | f \in S\} \cup    \{f|_{\alpha_{l}, \alpha_{1}=0} | f \in S\} \cup \\& \qquad \{\frac{\partial f_{1}|_{\alpha_{l}=0}}{\partial \alpha_{1}}f_{2}|_{\alpha_{l},\alpha_{1} = 0} - \frac{\partial f_{2}|_{\alpha_{l}=0}}{\partial \alpha_{1}}f_{1}|_{\alpha_{l},\alpha_{1}=0} | f_{1},f_{2} \in S \}) 
\\ &=
\text{irreducible factors of } (\{ \frac{\partial f}{\partial \alpha_{1}}|_{\alpha_{l}=0}  | f \in S\} \cup    \{f|_{\alpha_{l}, \alpha_{1}=0} | f \in S\} \cup \\& \qquad \{\frac{\partial f_{1}}{\partial \alpha_{1}}|_{\alpha_{l}=0}f_{2}|_{\alpha_{l},\alpha_{1} = 0} - \frac{\partial f_{2}}{\partial \alpha_{1}}|_{\alpha_{l}=0}f_{1}|_{\alpha_{l},\alpha_{1}=0} | f_{1},f_{2} \in S \})
\\&= \text{irreducible factors of } (\{ \frac{\partial f}{\partial \alpha_{1}}|_{\alpha_{l}=0}  | f \in S\} \cup    \{f|_{\alpha_{l}, \alpha_{1}=0} | f \in S\} \cup \\& \qquad \{(\frac{\partial f_{1}}{\partial \alpha_{1}}f_{2}|_{\alpha_{1} = 0} - \frac{\partial f_{2}}{\partial \alpha_{1}}f_{1}|_{\alpha_{1}=0})|_{\alpha_{l} =0} | f_{1},f_{2} \in S \}) \\
& \subseteq  \text{irreducible factors of } \{f|_{\alpha_{l}=0} | f \in \text{irreducible factors of } S^{4} \} \\
&= \text{irreducible factors of } S^{l'}_{[\alpha_{1}]}.
\end{align*}

By previous observations, the set of irreducible factors of $S^{l'}_{[\alpha_{1}]}$ is reducible with order $(\alpha_{2},\ldots,\alpha_{n})$, so by Lemma \ref{subsetlemma}, we get that  $S^{l}_{[\alpha_{1}]}$ is reducible with order $(\alpha_{2},\ldots,\alpha_{n})$, completing the proof. 
\end{proof}

Before we can prove the analogous statement for leading coefficients, we need to make a few more claims.

\begin{proposition}
\label{removingconstants}
Let $S = S' \cup \{c\}$ for some set of polynomials $S'$ and any constant $c$. Then $S$ is reducible if and only if $S'$ is reducible. 
\end{proposition}

\begin{proof}

If $S$ is reducible, then the result follows from Lemma \ref{subsetlemma}.

Now suppose $S'$ is a set of polynomials in the polynomial ring $\mathbb{Q}[\alpha_{1},\ldots,\alpha_{r}]$ which is reducible for some permutation $\sigma$ of $\alpha_{1},\ldots,\alpha_{r}$.  For $\sigma(1)$, the polynomial $c$ forces $c \in S^{2}$ and $0 \in S^{1}$. Thus the polynomial $g_{i}h_{j} - h_{i}g_{j} = cg_{j}$ when $f_{i} = c$. Then after factoring over $\mathbb{Q}$, we get $\tilde{S} = \{c, S'_{\sigma(1)}\}$. This implies in the compatibility graph, $c$ is adjacent to all other polynomials. Notice that this happens at every step of the algorithm. Then since $c$ is linear in every variable, by induction $S$ is  reducible.

\end{proof}

\begin{lemma}
\label{equationgrind}
Let $f_{1}$ and $f_{2}$ be polynomials linear in $\alpha_{1}$. Fix a variable $\alpha_{l}$ such that $\alpha_{l} \neq \alpha_{1}$. For a polynomial $f$, let $lc(f)$ be the leading coefficient of $f$ with respect to $\alpha_{l}$.  Then either
\begin{enumerate}
 \item{$\frac{\partial lc(f_{1})}{\partial \alpha_{1}}lc(f_{2})|_{\alpha_{1} = 0} - \frac{\partial lc(f_{2})}{\partial \alpha_{1}}lc(f_{1})|_{\alpha_{1}=0} = 0$ or}
 \item{ $\frac{\partial lc(f_{1})}{\partial \alpha_{1}}lc(f_{2})|_{\alpha_{1} = 0} - \frac{\partial lc(f_{2})}{\partial \alpha_{1}}lc(f_{1})|_{\alpha_{1}=0} = lc(\frac{\partial f_{1}}{\partial \alpha_{1}} f_{2}|_{\alpha_{1}=0} - \frac{\partial f_[2}{\partial \alpha_{1}} f_{1}|_{\alpha_{1} = 0}).$}
\end{enumerate}
\end{lemma}

\begin{proof}
Let $f_{1} = \sum_{i=0}^{k} (f_{1,i,1}\alpha_{1} + f_{1,i,2})\alpha^{i}_{l}$ and $f_{2} = \sum_{j=0}^{t} (f_{2,j,1}\alpha_{1} + f_{2,j,2})\alpha^{j}_{l}$. Then 
\begin{align*}
\frac{\partial lc(f_{1})}{\partial \alpha_{1}}lc(f_{2})|_{\alpha_{1} = 0} - \frac{\partial lc(f_{2})}{\partial \alpha_{1}}lc(f_{1})|_{\alpha_{1}=0} = f_{1,k,1}f_{2,t,2} - f_{2,t,1}f_{1,k,2},
\end{align*}
and 
\begin{align*}
lc(\frac{\partial f_{1}}{\partial \alpha_{1}} f_{2}|_{\alpha_{1}=0} - \frac{\partial f_{2}}{\partial \alpha_{1}} f_{1}|_{\alpha_{1} = 0}) =
lc((\sum_{i=0}^{k}f_{1,i,1}\alpha_{l}^{i})(\sum_{j=0}^{t}f_{2,j,2}\alpha_{l}^{j}) - (\sum_{j=0}^{t} f_{2,j,1}\alpha_{l}^{j})(\sum_{i=0}^{k}f_{1,i,2}\alpha_{l}^{i})).
\end{align*}

Now we consider various cases. If $f_{1,k,1}f_{2,t,2} - f_{2,t,1}f_{1,k,2} =0$, then the claim immediately holds. Therefore we assume that $f_{1,k,1}f_{2,t,2} - f_{2,t,1}f_{1,k,2} \neq 0$. Suppose  each of $f_{1,k,1},f_{2,t,2},f_{2,t,1}$ and $f_{1,k,2}$ are non-zero, then the largest power of $\alpha_{l}$ in $(\sum_{i=0}^{k}f_{1,i,1}\alpha_{l}^{i})(\sum_{j=0}^{t}f_{2,j,2}\alpha_{l}^{j}) - (\sum_{j=0}^{t} f_{2,j,1}\alpha_{l}^{j})(\sum_{i=0}^{k}f_{1,i,2}\alpha_{l}^{i})$ is $\alpha_{l}^{k+t}$, and thus 
\[lc((\sum_{i=0}^{k}f_{1,i,1}\alpha_{l}^{i})(\sum_{j=0}^{t}f_{2,j,2}\alpha_{l}^{j}) - (\sum_{j=0}^{t} f_{2,j,1}\alpha_{l}^{j})(\sum_{i=0}^{k}f_{1,i,2}\alpha_{l}^{i})) = f_{1,k,1}f_{2,t,2} - f_{2,t,1}f_{1,k,2}.\] 

Now suppose that $f_{1,k,1} =0$ and $f_{2,t,2}, f_{2,t,1},$ and $f_{1,k,2}$ are not zero. Then the highest power of $\alpha_{l}$ in $(\sum_{i=0}^{k}f_{1,i,1}\alpha_{l}^{i})(\sum_{j=0}^{t}f_{2,j,2}\alpha_{l}^{j}) - (\sum_{j=0}^{t} f_{2,j,1}\alpha_{l}^{j})(\sum_{i=0}^{k}f_{1,i,2}\alpha_{l}^{i})$ is $\alpha_{l}^{k+t}$ which only occurs in  $(\sum_{j=0}^{t} f_{2,j,1}\alpha_{l}^{j})(\sum_{i=0}^{k}f_{1,i,2}\alpha_{l}^{i})$ so,

\[lc((\sum_{i=0}^{k}f_{1,i,1}\alpha_{l}^{i})(\sum_{j=0}^{t}f_{2,j,2}\alpha_{l}^{j}) - (\sum_{j=0}^{t} f_{2,j,1}\alpha_{l}^{j})(\sum_{i=0}^{k}f_{1,i,2}\alpha_{l}^{i})) =  -f_{2,t,1}f_{1,k,2}.\]

Now suppose that $f_{1,k,2} = 0$ and $f_{2,t,2}, f_{2,t,1},$ and $f_{1,k,2}$ are not zero. Then the highest power of $\alpha_{l}$ in $(\sum_{i=0}^{k}f_{1,i,1}\alpha_{l}^{i})(\sum_{j=0}^{t}f_{2,j,2}\alpha_{l}^{j}) - (\sum_{j=0}^{t} f_{2,j,1}\alpha_{l}^{j})(\sum_{i=0}^{k}f_{1,i,2}\alpha_{l}^{i})$ is $\alpha_{l}^{k+t}$ so, 

\[lc((\sum_{i=0}^{k}f_{1,i,1}\alpha_{l}^{i})(\sum_{j=0}^{t}f_{2,j,2}\alpha_{l}^{j}) - (\sum_{j=0}^{t} f_{2,j,1}\alpha_{l}^{j})(\sum_{i=0}^{k}f_{1,i,2}\alpha_{l}^{i})) =  f_{1,k,1}f_{2,t,2}.\]

The case where $f_{2,t,1} =0$ and $f_{2,t,2}, f_{1,k,1},$ and $f_{1,k,2}$ are non-zero and the case where $f_{2,t,2} = 0$ and $f_{2,t,1}$, $f_{1,k,1},f_{1,k,2}$ are non-zero follow similarly. 

Now consider the case where $f_{1,k,1} = 0$, $f_{2,t,1} =0$, and $f_{1,k,2}$, $ f_{2,t,1}$ are not zero. Then $f_{1,k,1}f_{2,t,2} - f_{2,t,1}f_{1,k,2} =0$,  satisfying the claim.

Now consider the case where $f_{1,k,2} =0$,$f_{2,t,2} =0$, and $f_{1,k,1}$,$f_{2,t,1}$ are not zero. Then $f_{1,k,1}f_{2,t,2} - f_{2,t,1}f_{1,k,2} =0$, satisfying the claim. 

Now consider the case where $f_{1,k,1} =0$, $f_{2,t,2} =0$ and $f_{1,k,2}$, $f_{2,t,1}$ are non-zero. Then the highest power of $\alpha_{l}$ in $(\sum_{i=0}^{k}f_{1,i,1}\alpha_{l}^{i})(\sum_{j=0}^{t}f_{2,j,2}\alpha_{l}^{j}) - (\sum_{j=0}^{t} f_{2,j,1}\alpha_{l}^{j})(\sum_{i=0}^{k}f_{1,i,2}\alpha_{l}^{i})$ is $\alpha_{l}^{k+t}$ so, 

\[lc((\sum_{i=0}^{k}f_{1,i,1}\alpha_{l}^{i})(\sum_{j=0}^{t}f_{2,j,2}\alpha_{l}^{j}) - (\sum_{j=0}^{t} f_{2,j,1}\alpha_{l}^{j})(\sum_{i=0}^{k}f_{1,i,2}\alpha_{l}^{i})) =  -f_{2,t,1}f_{1,k,2}.\]

The case where $f_{1,k,2} = 0$, $f_{2,t,1}=0$ and $f_{1,k,1}$,$f_{2,t,1}$ are non-zero follows similarly. Notice that these are all possible cases, so the claim holds. 
\end{proof}

\begin{theorem}
\label{leadingterms}
Let $S = \{P_{1},\ldots,P_{N}\}$ be a set of polynomials which is reducible in the order $(\alpha_{1},\ldots,\alpha_{n})$. Fix $l \in \{1,\ldots,n\}$. For all $i \in \{1,\ldots,N\}$,  Let $S^{l} = \{lc(P)| P \in S\}$. Then $S^{l}$ is reducible with the order $(\alpha_{1},\ldots, \alpha_{n})$. 
\end{theorem}

\begin{proof}
Let $S$ and $S^{l}$ be a counterexample with $n$ minimized, as in, $S$ is reducible with order $(\alpha_{1},\ldots, \alpha_{n})$, but $S^{l}$ is not reducible with order $(\alpha_{1},\ldots, \alpha_{n})$

First we suppose that $l =1$. Then for all $P \in S$, we have $\deg(P, \alpha_{l}) \leq 1$. Then $S^{l} = S^{1}$ where $S^{1}$ is the set obtained in the reduction algorithm applied to $S$ and $\alpha_{l}$. Then $S^{l}_{[\alpha_{l}]} \subseteq S_{[\alpha_{l}]}$. Since $S_{[\alpha_{l}]}$ is reducible with order $(\alpha_{2},\ldots,\alpha_{n})$, by Lemma \ref{subsetlemma} we have that $S^{l}_{[\alpha_{l}]}$ is reducible with order $(\alpha_{2},\ldots,\alpha_{n})$, and thus $S^{l}$ is reducible with order $(\alpha_{1},\ldots,\alpha_{n})$.

Therefore we can assume that $l \neq 1$. Now consider the set $(S_{[\alpha_{1}]})^{l}$, which we define to be the set of leading coefficients of polynomials in $S_{[\alpha_{1}]}$ with respect to $\alpha_{l}$. Since $S$ is reducible with order $(\alpha_{1},\ldots,\alpha_{n})$, we have that $S_{[\alpha_{1}]}$ is reducible with order $(\alpha_{2},\ldots,\alpha_{n})$ and thus as we chose a minimal counterexample, $(S_{[\alpha_{1}]})^{l}$ is reducible with order $(\alpha_{2},\ldots, \alpha_{n})$.  Notice from the definitions we have:
\begin{align*}
(S_{[\alpha_{1}]})^{l} &= \{lc(f) | f \in S_{[\alpha_{1}]} \} 
\\&=
\{lc(f) | f \in \text{irreducible factors of } S^{4} \}.
\end{align*}

 For notational convenience, we will say $S^{jl}$ will be the set $S^{j}$ obtained by starting with $S^{l}$ for $j \in \{1,2,3,4\}$. Suppose we have a polynomial $f$, and $f = f_{1}f_{2}$ for some polynomials $f_{1}$ and $f_{2}$. Now, notice that $S^{l}_{[\alpha_{1}]}$ exists and,
\allowdisplaybreaks
\begin{align*}
S^{l}_{[\alpha_{1}]} 
&= \text{irreducible factors of } S^{4l}
 \\ &=
\text{irreducible factors of } (S^{1l} \cup S^{2l} \cup S^{3l}) 
\\ &=
\text{irreducible factors of } (\{ \frac{\partial f}{\partial \alpha_{1}} | f \in S^{l} \} \cup    \{f|_{\alpha_{1}=0} | f \in S^{l} \} \cup \\& \qquad \{\frac{\partial f_{1}}{\partial \alpha_{1}}f_{2}|_{\alpha_{1} = 0} - \frac{\partial f_{2}}{\partial \alpha_{1}}f_{1}|_{\alpha_{1}=0} | f_{1},f_{2} \in S^{l} \}) 
\\&=
\text{irreducible factors of } (\{ \frac{\partial lc(f)}{\partial \alpha_{1}}  | f \in S\} \cup    \{lc(f)|_{\alpha_{1}=0} | f \in S\} \cup \\& \qquad \{\frac{\partial lc(f_{1})}{\partial \alpha_{1}}lc(f_{2})|_{\alpha_{1} = 0} - \frac{\partial lc(f_{2})}{\partial \alpha_{1}}lc(f_{1})|_{\alpha_{1}=0} | f_{1},f_{2} \in S \}) 
\\& \subseteq \text{irreducible factors of } (\{ lc(\frac{\partial f}{\partial \alpha_{1}})  | f \in S\} \cup    \{lc(f|_{\alpha_{1}=0}) | f \in S\} \cup \\& \qquad \{lc(\frac{\partial f_{1}}{\partial \alpha_{1}}f_{2}|_{\alpha_{1} = 0} - \frac{\partial f_{2}}{\partial \alpha_{1}}f_{1}|_{\alpha_{1}=0}) | f_{1},f_{2} \in S \}) \\
&=  \text{irreducible factors of } \{lc(f) | f \in \text{irreducible factors of } S^{4} \}  \\
&= \text{irreducible factors of } (S_{[\alpha_{1}]})^{l}
\end{align*}

We note that the subset relationship between lines $4$ and $5$ above holds by appealing to Lemma \ref{equationgrind} and assuming that if $\frac{\partial lc(f_{1})}{\partial \alpha_{1}}lc(f_{2})|_{\alpha_{1} = 0} - \frac{\partial lc(f_{2})}{\partial \alpha_{1}}lc(f_{1})|_{\alpha_{1}=0} = 0$, then we remove the $0$ from the set. This does not affect reducibility since $0$ is a constant (see Proposition \ref{removingconstants}).

Then since $(S_{[\alpha_{1}]})^{l}$ is reducible with order $(\alpha_{2},\ldots,\alpha_{n})$, the set of irreducible factors of $(S_{[\alpha_{1}]})^{l}$ is reducible with order $(\alpha_{2},\ldots,\alpha_{n})$ by Lemma \ref{irreduciblepolynomials}. Then by Lemma \ref{subsetlemma},  we have that $S^{l}_{[\alpha_{1}]}$ is reducible with order $(\alpha_{2},\ldots,\alpha_{n})$, completing the claim. 
\end{proof}

With that, Theorem \ref{fullclaim} is proved. Now we shift focus to reducibility for the specific set $S = \{\Psi_{G}, \Phi_{G}\}$ for some graph $G$.

\section{Graph Minors and Reducibility}\label{sec minors}
Given a graph $G$, and an edge $e$, we denote the graph obtained by deleting $e$ as $G \setminus e$ and the graph obtained by contracting $e$ as $G / e$. A graph $G$ has a graph $H$ as a \textit{minor} if $H$ can be obtained from $G$ via a series of edge deletions, contractions and if necessary, deletion of isolated vertices. We remark that if $G$ is connected, then any $H$-minor of $G$ can be obtained by deleting and contracting edges only.

We now prove that reducibility for the Symanzik polynomials is graph minor closed for graphs with a fixed number of momenta. Notice that $\Psi_{G}$ is a homogeneous polynomial linear in all variables, and if there are no massive edges, then $\Phi_{G}$ is a homogeneous polynomial linear in all variables \cite{Feynmangraphpolynomials}. We also note the following well known observations (see, for example, \cite{Bognerpaper, Feynmangraphpolynomials}):

\begin{lemma} 
\label{derivativeandevaluationsetrelationship}
Let $G$ be a graph. Let $e \in E(G)$ such that $e$ is not massive and $e$ is not a loop. Then the following identities hold:

\[\frac{\partial}{\partial \alpha_{e}}\Phi_{G} = \Phi_{G\setminus e}, \ \Phi_{G}|_{\alpha_{e} = 0} = \Phi_{G / e},\]

\[\frac{\partial}{\partial\alpha_{e}}\Psi_{G} = \Psi_{G\setminus e}, \  \Psi_{G}|_{\alpha_{e} = 0} = \Psi_{G / e}.\]

\end{lemma}

\begin{lemma}
\label{loopedgesSymanzikPolys}
Let $G$ be a graph and consider any edge $e$ which is a loop. Then $\Psi_{G} = \alpha_{e}\Psi_{G \setminus e}$, and $\Phi_{G} = \Phi_{G \setminus e}\alpha_{e}$.
\end{lemma}

Before proving the graph minor closed result, we make the convention that for a given graph $G$, we consider reducibility in the polynomial ring $\mathbb{Q}[\alpha_{1},\ldots,\alpha_{r}]$ where each $\alpha_{i}$ corresponds to a Schwinger parameter of an edge of $G$. However using this convention means that we cannot appeal to Theorem \ref{fullclaim} as it is stated, as deleting and contracting edges would reduce the number of variables we are considering, whereas in Theorem \ref{fullclaim}, the variable list never changes. 
However, the corresponding result with $\alpha_l$ removed from the variable list after taking constant or leading coefficients is an immediate corollary because, as observed previously, a reduction step where the variable does not appear makes no change to the polynomial while potentially adding edges to the compatibility graph which can only make it harder to be reducible.
Now we prove the graph minor closed result.  

\begin{theorem}
\label{MinorClosed}
Let $r$ be a fixed positive integer. Let $\mathcal{G}$ be the set of connected graphs which have $r$ external momenta, and are reducible with respect to $\{\Phi,\Psi\}$. Then all connected minors of $G$ are in $\mathcal{G}$. 
\end{theorem}

\begin{proof}
Let $G \in \mathcal{G}$. Pick any edge $e \in E(G)$ (if there are no edges then the result follows trivially). If $e$ is massive, then $G$ is reducible for any value of $m_{e}$ and so we can suppose $m_{e}=0$. Therefore without loss of generality we assume $e$ is massless.  We consider two cases.

 \textbf{Case 1:} Suppose $e$ is a loop. Graph theoretically for loops, deletion and contraction are the same. Then by Lemma \ref{loopedgesSymanzikPolys}, we have $\{\Phi_{G}, \Psi_{G}\} = \{\Phi_{G \setminus e}\alpha_{e}, \Psi_{G \setminus e}\alpha_{e}\}$. By Lemma \ref{irreduciblepolynomials}, $\{\Phi_{G \setminus e}\alpha_{e}, \Psi_{G \setminus e}\alpha_{e}\}$ is reducible if and only if $\{\Phi_{G \setminus e}, \Psi_{G \setminus e}, \alpha_{e}\}$ is reducible. Since $\alpha_{e}$ is a monomial, we get that $\{\Phi_{G}, \Psi_{G}\}$ is reducible if and only if $\{\Phi_{G \setminus e}, \Psi_{G \setminus e}\}$ is reducible. Alternately if we take the more physical convention that contracting a loop sends the whole graph to $0$ (corresponding to the fact that $\Psi_G|_{\alpha_{e}=0} = 0$ and $\Phi_G|_{\alpha_{e}=0}=0$ for $e$ a loop) then the result still holds as $\{0\}$ is trivially reducible.  Therefore $G \setminus e \in \mathcal{G}$, and by induction the result follows. 

\textbf{Case 2:} Suppose that $e$ is not a loop edge and consider $G \setminus e$ and $G / e$. By Lemma \ref{derivativeandevaluationsetrelationship} the Symanzik polynomials for $G \setminus e$ are $\{ \frac{\partial}{\partial\alpha_{e}}\Phi_{G}, \frac{\partial}{\partial\alpha_{e}}\Psi_{G} \}$. Similarly, by Lemma \ref{derivativeandevaluationsetrelationship} the Symanzik polynomials for $G /e$ are $\{\Psi_{G|_{\alpha_{e} = 0}}, \Phi_{G|_{\alpha_{e}=0}}\}$. Notice that since we have no massive edges, both of the Symanzik polynomials are linear in $\alpha_{l}$, thus $\frac{\partial}{\partial \alpha_{e}} \Phi_{G} = lc(\Phi_{G})$ and $\frac{\partial}{\partial \alpha_{e}} \Phi_{G} = lc(\Phi_{G})$. Thus both sets are reducible by appealing to Theorem \ref{leadingterms} and Theorem \ref{reducibilityevalutation}. Therefore both $G \setminus e$ and $G /e$ are in $\mathcal{G}$, and by induction the result follows. 
\end{proof}

We note the exact same proof works for reducibility with respect to just $\Phi$ or $\Psi$.

\subsection{Towards a Forbidden Minor Characterization}

\begin{figure}
\begin{center}
\includegraphics[scale=0.75]{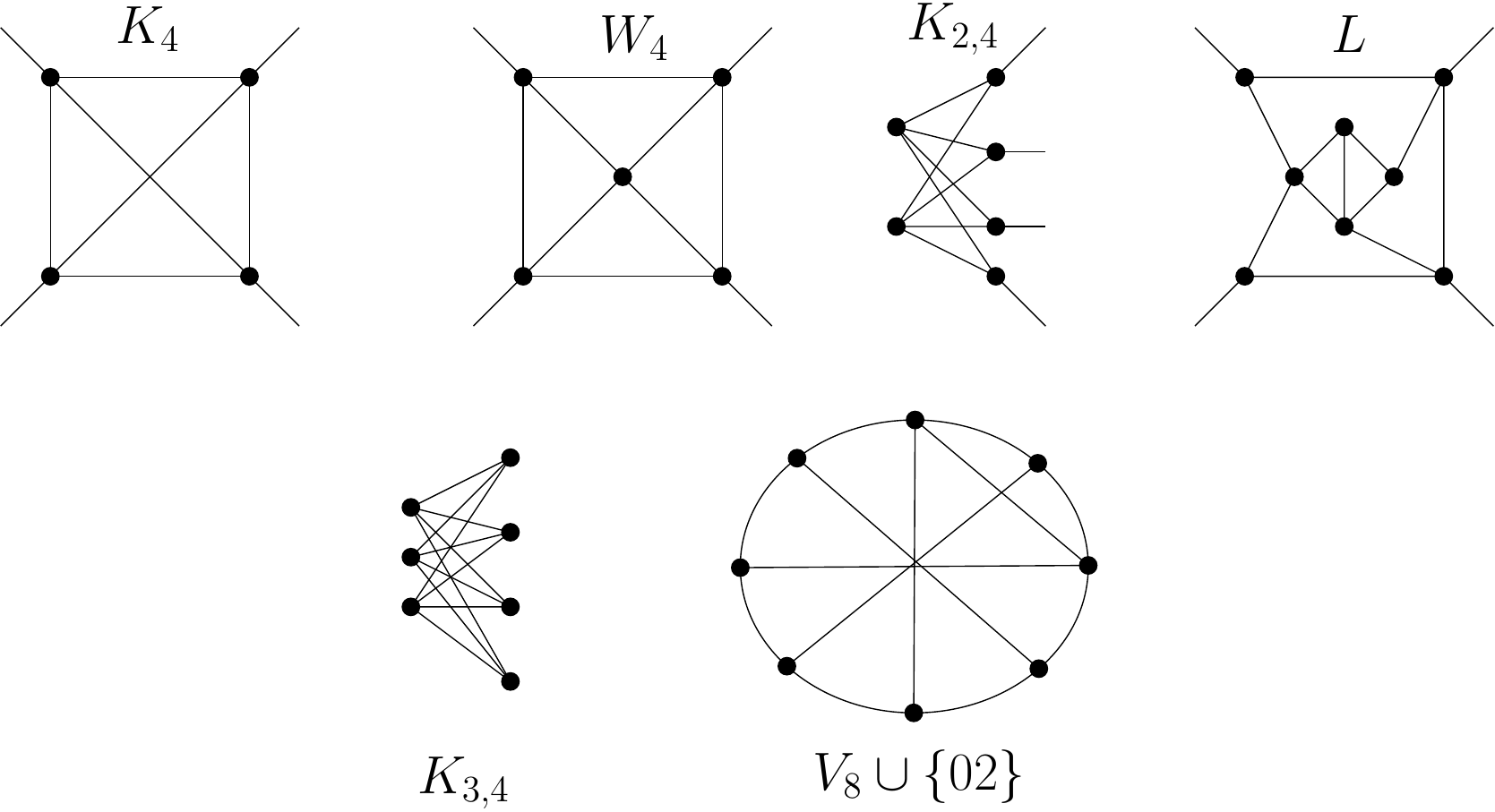}
\caption{Forbidden Minors for reducibility. Edges with no endpoint denote vertices with external momenta.}
\label{forbiddenminors}
\end{center}
\end{figure}

As mentioned in the introduction, Robertson and Seymour proved that if a property is graph minor closed, then that property is characterized by a finite set of forbidden minors (where a graph is a forbidden minor if it does not have the property, but all proper minors of it have the property) \cite{WQOtheorem}. It is worth noting that even restricted to connected graphs this theorem holds and that appropriate analogues hold for graphs with external momenta (rooted minor results) as will be discussed below.  Then by Theorem \ref{MinorClosed}, reducibility with a fixed number of external momenta is characterized by a finite set of forbidden minors. Therefore, it is of interest to try and determine the full set of forbidden minors.  We do not claim to have anywhere near a complete characterization, but we list some of the graphs which are known to be non-reducible here. Checking whether a graph is reducible was done using Panzer's HyperInt program \cite{panzerphdthesis, HyperintArticle}. We note that Bogner also has software available for checking reducibility \cite{MPLarticle}. 

Quite a bit of work has been done on reducibility with respect to the first Symanzik polynomial. In \cite{FrancisBig}, the graph $K_{3,4}$ was shown to be non-reducible with respect to $\Psi$ and that, in particular, for any permutation $\sigma$, the set $S_{[\sigma(1),\ldots,\sigma(7)]}$ (if it exists) is quadratic in $\sigma(8)$. Furthermore, by direct computation, one can show that all proper minors of $K_{3,4}$ are reducible, showing that $K_{3,4}$ is a forbidden minor for reducibility with respect to the first Symanzik polynomial \cite{FrancisBig}. 

Also in \cite{FrancisBig}, Brown notes that the graph $V_{8} \cup \{02\}$ (see Figure \ref{forbiddenminors}, also known as $Q48$ in \cite{olivercensus}, $V_{8}$ is the well known Wagner's graph) is not reducible with respect to $\Psi$. Again by direct computation, it is possible to show that $V_{8} \cup \{02\}$ is a forbidden minor with respect to $\Psi$ (\cite{FrancisBig}). 

Now we record forbidden minors for reducibility with respect to both Symanzik polynomials, where the underlying graph has no massive edges, and four on-shell external momenta (so $\rho^{2} =0$ for all external momenta $\rho$). It was shown in \cite{Martin} and noted in \cite{Bognerpaper} that the complete graph on four vertices, $K_{4}$, is a forbidden minor, where each vertex has an external edge. Also, the wheel with four spokes, $W_{4}$ is a forbidden minor when all of the vertices with external edges lie on the four vertices with degree $3$. We note that $W_{4}$ is the forbidden minor obtained when taking the non-reducible graphs in \cite{Martin} and finding the minimal non-reducible graph. One can show that $K_{2,4}$ where all the external edges are on the large side of the bipartition is a forbidden minor. Lastly, there is a graph which we call $L$ which is a new forbidden minor (see Figure \ref{forbiddenminors}). To the authors knowledge, these are all known forbidden minors, though we expect there to be more which are not yet known. 

On finding more forbidden minors, a census of periods of Feynman integrals has been constructed in \cite{olivercensus}. Since reducibility is tied to Brown's integration algorithm, and we know that when Brown's algorithm succeeds that the Feynman integral is a multiple zeta value, all graphs in \cite{olivercensus} whose Feynman integral is not a multiple zeta value are not reducible. There are some graphs which are known to not have periods which are multiple zeta values, but they are almost certainly not forbidden minors. However, it is quite difficult computationally to try and find the forbidden minor from these graphs, and it would be beneficial to find a easier way to compute reducibility. Additionally, we have the following conjecture.

\begin{Conjecture}
\label{k6minor}
The complete graph on six vertices, $K_{6}$, is a forbidden minor for reducibility with respect to the first Symanzik polynomial. 
\end{Conjecture} 

We remark that it is not known if $K_{6}$ is even not reducible. It turns out that the resolution of this conjecture will largely not impact the class of reducible graphs with respect to the first Symanzik polynomial. This is due to the following observation:

\begin{observation}
\label{K34andK6}
Let $G$ be a $3$-connected graph with a $K_{6}$-minor. Then either $G$ is isomorphic to $K_{6}$, or $G$ has a $K_{3,4}$-minor. 
\end{observation}

This follows easily from Seymour's $3$-splitter Theorem \cite{splittertheorem} (see \cite{Ben} for a proof of the observation). We note, despite not having a full list of forbidden minors, one can utilize the known forbidden minors to gain insight on the structure of reducible graphs.  Here we survey some of the known results, in particular results dealing with the four external on-shell momenta case.

First, we have to extend the notion of graph minors to deal with the external momenta. For the purposes of this discussion, the \textit{roots} of a graph $G$ will be the the vertices of $G$ with external momenta. Let $X$ be the set of roots of a graph. Then, given graphs $G$ and $H$, and a injective map $f:X \to V(H)$, we say $G$ admits a \textit{rooted $H$-minor with respect to $f$} if there exists a set $\{G_{x} \ | \ x \in V(H)\}$ where $G_{x}$ is a set of vertices, we have $G_{x} \cap G_{y} = \emptyset$ when $x \neq y$, for each $uv \in E(H)$, there is a $x \in G_{u}$ and $y \in G_{v}$ such that $xy \in E(G)$, and lastly if $x \in X$, then $x \in G_{f(x)}$. We note that if one drops the final condition, then one obtains an alternative definition of graph minors. In general, we will want to consider more than one injective map, so if we have injective maps $f_{1},\ldots,f_{n}: X \to V(H)$, that a graph $G$ has a \textit{rooted $H$-minor} if $G$ has a rooted $H$-minor with respect to $f_{i}$ for any $i \in \{1,\ldots,n\}$. 

For our purposes, for each of $K_{4}$, $W_{4}$, $K_{2,4}$ and $L$, we consider all injective maps from the roots of $G$ to the external momenta of those graphs as shown in Figure \ref{forbiddenminors}. We note that the notion of rooted minors is the correct notion of minors when dealing with reducibility with respect to both Symanzik polynomials. If a graph has a rooted $K_{4}$, $W_{4}$, $K_{2,4}$ or $L$-minor, then the graph is not reducible with respect to $\{\Phi,\Psi\}$. Then we can ask the question; what is the structure of a graph not containing rooted $K_{4}$, $W_{4}$, $K_{2,4}$, or $L$-minors?  There has been some work on varying aspects of this question \cite{Ben, Graphminorspaper, root, linodemasithesis}. It turns out for $3$-connected graphs, one of the those minors always exists. In fact, we have:

\begin{theorem}[\cite{Ben, Graphminorspaper}]
Let $G$ be a $3$-connected graph with four roots. Then either $G$ has a rooted $K_{4}$-minor or a rooted $W_{4}$-minor. 
\end{theorem}

Therefore, reducibility with four external on-shell momenta is only interesting for graphs with cut vertices or $2$-separations.  This reduction already simplifies the physics significantly.   As all the minors we are looking at are $2$-connected, graphs with cut vertices are uninteresting, so we restrict to $2$-connected graphs. It turns out that to not have one of the forbidden minors, one has to look largely planar. To formalize this, let $H$ be a planar graph where the outer face is bounded by a $4$-cycle on vertices $a,b,c$ and $d$   , all internal faces are triangles, and all triangles are faces. To each triangle $T$ in $H$, we attach a clique, $F_{T}$, of arbitrary (possibly empty) size where the vertices of $F_{T}$ are only adjacent to to the vertices in $T$ in $H$. Following Wood and Fabila-Monroy, \cite{root}, we call the resulting graph an \textit{$\{a,b,c,d\}$-web}\footnote{Note that this is a different notion of webs than used in studying infrared divergences in non-abelian gauge theories \cite{Wweb} or in studying the Grassmannian via tensor invariants and other similar objects \cite{Kweb, FLLweb}.}. In general, given a graph $H$, we say the graph $H^{+}$ is obtained by adding cliques of arbitrary size to every triangle in $H$ in the same fashion as above. We remark that webs are closely related to the notion of \textit{flat embeddings}, which arise in the study of graph linkages.
Now we state the excluded rooted $K_{4}$-minor theorem:

\begin{figure}
\centering
\includegraphics[scale =0.6]{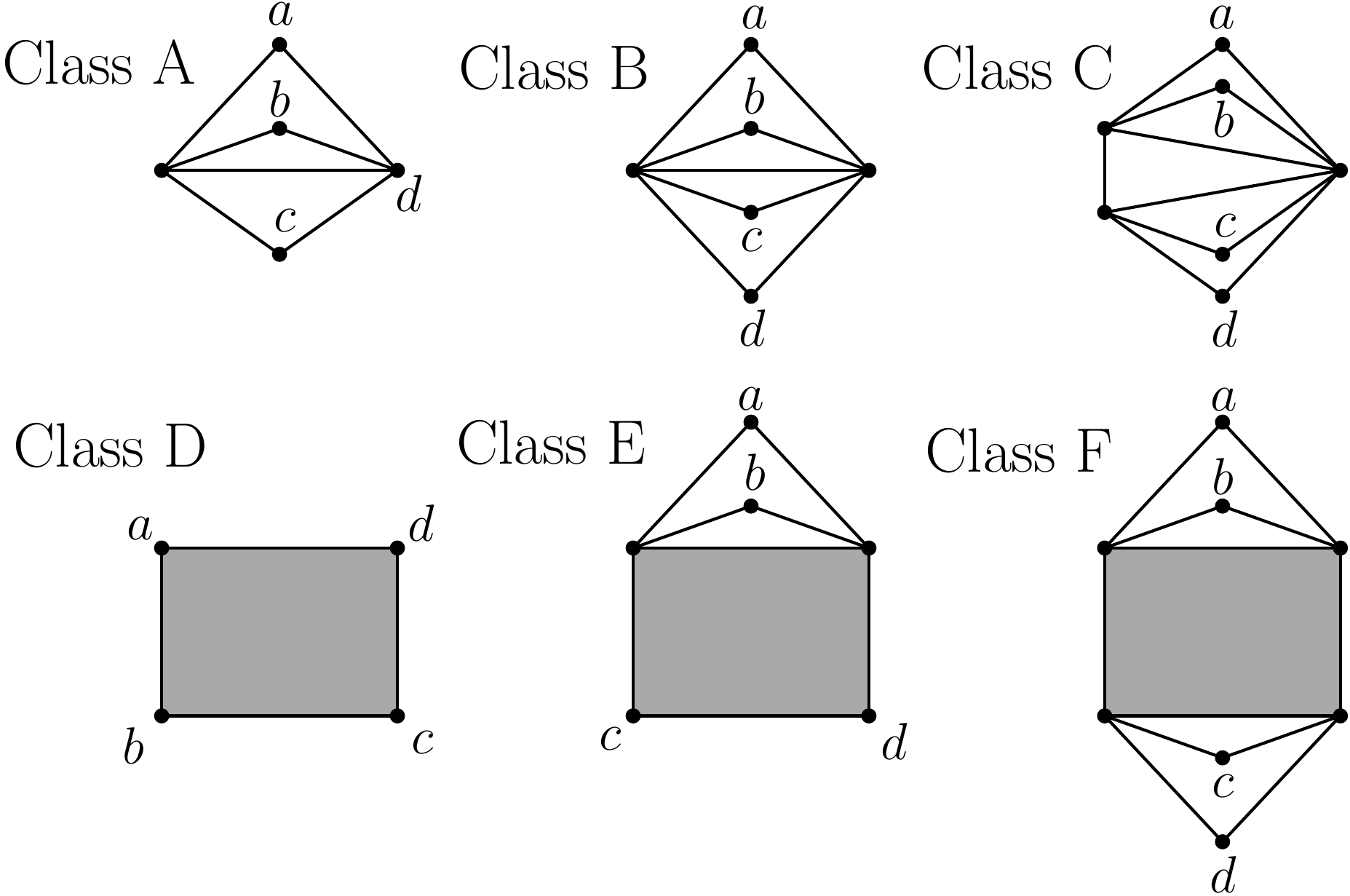}
\caption{The graphs $H$ in Theorem \ref{K4free}. Grey sections are webs in the sense of Wood and Fabila-Monroy. The vertices labelled $a,b,c,d$ are the vertices with external momentum.}
\label{k4freefigure}
\end{figure}

\begin{theorem}[\cite{root}]
\label{K4free}
Suppose $G$ has four roots and is rooted $K_{4}$-minor-free. Then $G$ is the spanning subgraph of a graph $H^{+}$ where $H$ is one of the graphs in Figure \ref{k4freefigure}. 
\end{theorem} 

 When we add in rooted $K_{2,4}$-minors and $W_{4}$-minors we can restrict the graph class quite significantly. For the rest of the paper, we will use the phrase $G$ is a spanning subgraph of a class $\mathcal{A}$ (for example) to mean that $G$ is isomorphic to a spanning subgraph of a graph $H^{+}$ where $H$ is the graph in Figure \ref{k4freefigure} under the header class $\mathcal{A}$. 

We can collect together and summarize the main forbidden minor characterization results of \cite{Ben} and \cite{Graphminorspaper} (which can be consulted for details) in the following theorem.
 
\begin{theorem}[\cite{Ben}, \cite{Graphminorspaper}]
Let $G$ be a $2$-connected graph with four roots. If $G$ is a spanning subgraph of a class $\mathcal{A}$ graph then $G$ does not have a rooted $W_{4}$ or $K_{2,4}$-minor.  If $G$ is a spanning subgraph of a class $\mathcal{B}$ or $\mathcal{C}$ graph, then $G$ has a rooted $K_{2,4}$-minor. If $G$ is a spanning subgraph of an $\{a,b,c,d\}$-web and has a rooted $K_{2,4}$-minor, then $G$ has a rooted $W_{4}$-minor. Furthermore, a spanning subgraph of an $\{a,b,c,d\}$-web does not have a rooted $W_{4}$-minor if and only if it is a spanning subgraph of one of the graphs in Figure \ref{Wheelobstructions} (see below for a more detail). Finally, spanning subgraphs of class $\mathcal{E}$ and $\mathcal{F}$ graphs have a rooted $W_{4}$-minor if and only if the underlying $\{a,b,c,d\}$-web has a rooted $W_{4}$-minor, where we move the roots which are not on the web to distinct neighbouring vertices on the web. 
\end{theorem}

\begin{figure}
\centering
\includegraphics[scale =0.5]{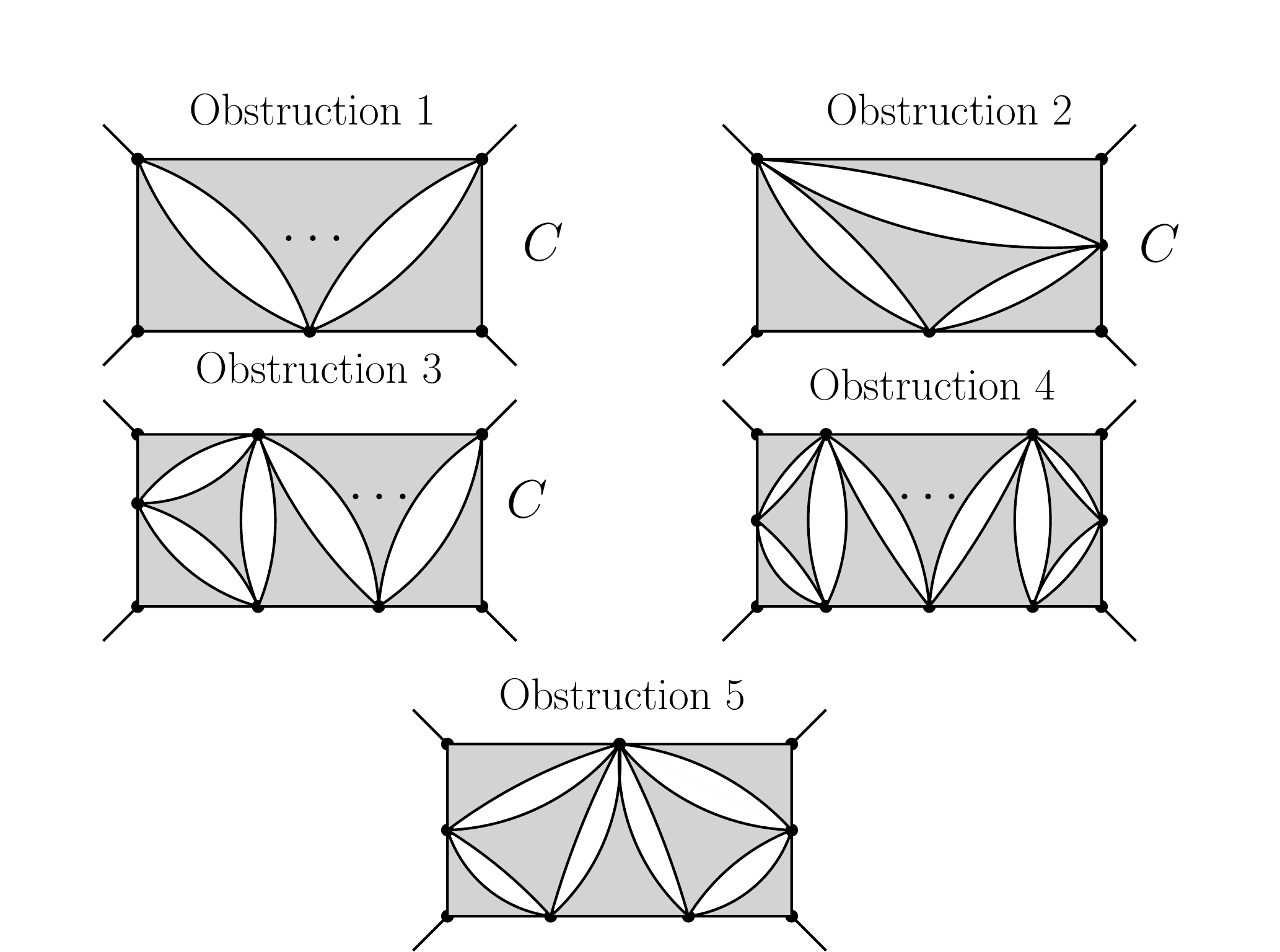}
\caption{Obstructions to having a rooted $W_{4}$-minor in an $\{a,b,c,d\}$-web. The curved white sections indicate a $2$-vertex cut. The dots indicate a $2$-chain. }
\label{Wheelobstructions}
\end{figure}

Now we clarify the obstructions in Figure \ref{Wheelobstructions}. All of the obstructions are build up of two pieces, \textit{root separating triangles} and  \textit{root separating $2$-chains}. A \textit{$2$-separation} $(A,B)$ is the partition of $V(G)$ induced by a $2$-vertex cut. A \textit{$2$-dissection} is a sequence of $2$-separations $((A_{1},B_{1}),\ldots,(A_{k},B_{k}))$, where $A_{i} \subseteq A_{i+1}$ for all $i \in \{1,\ldots,k-1\}$ and $B_{i+1} \subseteq B_{i}$ for all $i \in \{2,\ldots,k\}$. A \textit{$2$-chain} is a $2$-dissection $((A_{1},B_{1}),\ldots,(A_{k},B_{k}))$ where $A_{i} \cap B_{i} \cap A_{i+1} \cap B_{i+1}$ is non-empty for all $i \in \{1,\ldots,k-1\}$. Given a graph with four roots, say $a,b,c,d$, a \textit{root separating $2$-chain} is a $2$-chain where $a \in A_{1} \cap B_{1}$, $b \in A_{1} \setminus B_{1}$, $c \in A_{k} \cap B_{k}$, and $d \in B_{k} \setminus A_{k}$. In Figure \ref{Wheelobstructions}, obstruction one is a root separating $2$-chain, and in general, any root separating $2$-chain is an obstruction. 

Now, we define a \textit{triangle} to be a collection of $2$-separations $(A_{1},B_{1}),(A_{2},B_{2}),(A_{3},B_{3})$ where $A_{1} \cap B_{1} = \{x,y\}$, $A_{2} \cap B_{2} = \{y,z\}$, and $A_{3} \cap B_{3} = \{x,z\}$. For notational convenience, we will enforce that in a triangle, $(A_{i} \setminus B_{i}) \cap (A_{j} \setminus B_{j}) = \emptyset$ for any $i,j \in \{1,2,3\}$, $i \neq j$. Then, a triangle $(A_{1},B_{1}),(A_{2},B_{2}),(A_{3},B_{3})$, is a  \textit{root separating triangle} if exactly two of the roots are contained in $A_{1}$, exactly one root is contained in $A_{2} \setminus B_{2}$ and exactly one root contained in $A_{3} \setminus B_{3}$. Obstruction $2$, $3$, $4$, and $5$ from Figure \ref{Wheelobstructions} are all built up from root separating triangles.

The only rooted minor which we haven't discussed is the rooted $L$-minor. It turns out that by additionally restricting to graphs without rooted $L$-minors we do not actually see that much change in the graph class. After applying some reductions which we won't discuss here, it turns out root separating $2$-chain obstructions are the only place where one can find rooted $L$-minors. In these instances, the only time one can only obtain a rooted $L$-minor if the root separating $2$-chain has even length, and ``inside" a pair of separations (as in, in the graph induced by $G[A_{i+1}] \cap G[B_{i}]$) one can find a certain smaller rooted minor. The parity condition arises from the fact that if the length of the two chain is odd, then for any cycle containing the roots $a,b,c,d$, the two roots not appearing in $A_{1} \cap B_{1}$ and $A_{k} \cap B_{k}$ are separated by the two roots which do appear in $A_{1} \cap B_{1}$ and $A_{k} \cap B_{k}$. As in, if $b,d \not \in A_{1} \cap B_{1}$ and $b,d \not \in A_{k} \cap B_{k}$, then walking in any direction along the cycle from $b$, we hit $a$ or $c$ before $d$.  We refer the reader to \cite{Ben, Graphminorspaper} for more details.

Notice throughout these results, nothing has been said about the arbitrary graph hidden in each triangle. This is because the rooted minors cannot ``see" the structure of these graphs, as the triangles are too small of a cut-set.  However, the forbidden minors for the first Symanzik polynomial (which also are obstructions for reducibility for both Symanzik polynomials) are not rooted, and thus impose structure on those sections of the graph. 

One small thing that can be said about the graphs in each triangle while still using the rooted minors is that if two triangles $T_{1}$ and $T_{2}$ have three disjoint paths between them, then at most one of the spanning subgraphs of cliques adjoined to the triangles is non-planar. This follows easily from a well known result on rooted $K_{2,3}$-minors by Seymour and Thomasson \cite{rootedk23theorem}. 

As to the effect of the known minors for the first Symanzik polynomial on the graphs hidden in each triangle, unfortunately, in general not much is known about graphs without $K_{3,4}$-minors and $V_{8} \cup \{02\}$-minors. In the projective plane, vertically $4$-connected graphs without $K_{3,4}$-minors have been characterized as graphs which are minors of M\"obius hyperladders (or spanning subgraphs of certain ``patch graphs") \cite{K34minorsprojpart1, K34minorsprojpart2}. Essentially nothing is known about $V_{8} \cup \{02\}$-minors, however $V_{8}$-minor-free graphs have been characterized by Robertson and Maharray \cite{NoV8minors}, so one might hope to be able to characterize $V_{8} \cup \{02\}$-minor free graphs by using the $V_{8}$ result and the splitter theorem. 

Despite not having a clear picture of what the class of reducible graphs for the first Symanzik polynomial looks like, one can make a reasonable guess on what the class should be. As mentioned in the introduction, Brown showed that graphs with vertex width less than or equal to three are reducible for the first Symanzik polynomial. Define vertex width as follows.

 Let $G$ be a graph with $m$ edges. The \textit{width} of an ordering $e_1, e_2,\ldots, e_m$ of $E(G)$ is the maximum order of a separation of the form $(\{e_1,\ldots,e_l\},\{e_{l+1},...,e_m\})$ for $l \in \{1,\ldots,m\}$. Here we are viewing $\{e_1, \ldots,e_l\}$ as a set of vertices induced by the edges $e_{1},\ldots, e_{l}$, and the order of the separation is $|\{e_1,\ldots,e_l\} \cap \{e_{l+1},...,e_m\}|$. The \textit{vertex width} of $G$ is the minimum width among all edge orders of $G$.
 
 We note that vertex width is closely related to other graph theoretic width properties. In particular, vertex width is close to the linear version of treewidth, pathwidth, and the linear notion of branch width, which is sometimes called caterpillar width.  In the 3-connected case vertex-width 3 also matches with zero forcing number equal to 3 \cite{AIMgraphs}.

It is easy to see that the graphs satisfying vertex width $\leq 3$ are planar. In fact, the $3$-connected graphs satisfying vertex width $\leq 3$ were characterized in terms of forbidden minors ($K_{5},K_{3,3}$, the cube, the octahedron, and one other special graph) by Crump in his master's thesis \cite{CrumpMastersthesis, Crumppaper}. Unfortunately, all of the forbidden minors for $3$-connected vertex width $\leq 3$ graphs are reducible with respect to the first Symanzik polynomial, so the class of reducible graphs for the first Symanzik polynomial should be somewhat larger, but it still gives a nice infinite family of graphs which are reducible for the first Symanzik polynomial. On the flip side, there is a beautiful conjecture of J\"orgensen which says that a $6$-connected graph is $K_{6}$-minor-free if and only if there exists a vertex $v$ such that the deletion of $v$ results in a planar graph. This conjecture has been proven for graphs with a large number of vertices \cite{jorgensensconjecture}. Assuming the truth of this conjecture (or assuming we are dealing with large graphs), combined with Observation \ref{K34andK6}, we see that graphs which are reducible should be ``nearly" planar, in the sense that deleting a small number of vertices leaves a planar graph. Thus we see that the class of reducible graphs for the first Symanzik polynomial should lie somewhere between the class of vertex width $\leq 3$, which are highly structured planar graphs, and graphs which are ``almost" planar. We also note that there is a planar graph which is not reducible \cite{nonmultiplezetavalue}, and so perhaps it is reasonable to expect the class of reducible class lies closer to highly structured vertex width $\leq 3$ graphs than the class of ``almost" planar graphs.

For reducibility of graphs with $4$ external on-shell momenta, given that going from rooted $K_{4}$, $W_{4}$, and $K_{2,4}$-minor free graphs to rooted $K_{4}$, $W_{4}$, $K_{2,4}$ and $L$-minor free graphs does not change the characterization too much, we anticipate the graph class in the exact characterization not to differ too much from what was presented here (however there may be significantly many more forbidden minors).

\bibliography{ThesisBib}
  \bibliographystyle{plain}
\end{document}